\documentclass[a4paper,runningheads]{llncs}

\usepackage{amsmath}
\usepackage{amssymb}
\usepackage{graphicx}
\usepackage{xspace}
\usepackage{enumerate}
\usepackage{comment}
\usepackage{cite}
\usepackage{algorithm}
\usepackage{times}
\usepackage{subfig}
\usepackage[noend]{algpseudocode}
\usepackage[usenames]{xcolor}
\usepackage{stmaryrd} 
\usepackage{mathtools} 

\spnewtheorem{prop}{Property}{\bfseries}{\itshape} 

\newcommand{\remove}[1]{}

\renewenvironment{proof}
{{\em Proof.\ }}{\hspace*{\fill}$\Box$\par\vspace{2mm}}

\renewcommand{\epsilon}{\varepsilon}
\newcommand{\E}{\mathcal{E}}
\newcommand{\seg}[1]{\overline{#1}}

\begin{document}
\title{Simultaneous Embeddings with \\Few Bends and Crossings 
  \thanks{Research of M.~Hoffmann and V.~Kusters is partially supported by
    the ESF EUROCORES programme EuroGIGA, CRP GraDR and the Swiss National
    Science Foundation, SNF Project 20GG21-134306.}
}
\author{Fabrizio Frati\inst{1} \and Michael Hoffmann\inst{2} \and Vincent Kusters\inst{2}}
\institute{Dipartimento di Ingegneria, University Roma Tre, Italy\\
\email{frati@dia.uniroma3.it}
\and Department of Computer Science, ETH Z\"urich, Switzerland\\
\email {\{hoffmann,vincent.kusters\}@inf.ethz.ch}}
\maketitle

\begin{abstract}
A \emph{simultaneous embedding with fixed edges} ({\sc Sefe}) of two planar graphs $R$ and $B$ 
is a pair of plane drawings of $R$ and $B$ that coincide when restricted to the common vertices and edges of $R$ and $B$. We show that whenever $R$ and $B$ admit a {\sc Sefe}, they also admit a {\sc Sefe} in which every edge is a polygonal curve with few bends and every pair of edges has few crossings. Specifically: (1) if $R$ and $B$ are trees then one bend per edge and four crossings per edge pair 
suffice (and one bend per edge is sometimes necessary), (2) if $R$ is a planar graph and $B$ is a tree then six bends per edge and eight crossings per edge pair suffice, and (3) if $R$ and $B$ are planar graphs then six bends per edge and sixteen crossings per edge pair suffice. Our results improve on a paper by Grilli et al{.} (GD'14), which proves that nine bends per edge suffice, and on a paper by Chan et al{.} (GD'14), which proves that twenty-four crossings per edge pair suffice. 
\end{abstract}


\section{Introduction}
\label{se:introduction}

Let $R=(V_R,E_R)$ and $B=(V_B,E_B)$ be two planar graphs sharing a \emph{common graph} $C=(V_R\cap V_B, E_R\cap E_B)$. The vertices and edges of $C$ are {\em common}, while the other vertices and edges are {\em exclusive}. We refer to the edges of $R$, $B$, and $C$ as the {\em red}, {\em blue}, and {\em black edges}, respectively. A \emph{simultaneous embedding} of $R$ and $B$ is a pair 
of plane drawings of $R$ and $B$, respectively, 
that agree on the common vertices (see Figs.~\ref{fig:se-sefe-sge-example-g1}--\ref{fig:se-sefe-sge-example-se}).


Simultaneous graph embeddings have been a central topic of investigation for the graph drawing community in the last decade, because of their applicability to the visualization of dynamic graphs and of multiple graphs on the same vertex set~\cite{bcdeeiklm-spge-07,ekln-sgd-05}, and because of the depth and breadth of the theory they have been found to be related to. 

\begin{figure}[t]
  \centering
  \subfloat[]{\label{fig:se-sefe-sge-example-g1}\includegraphics{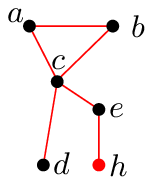}}\hfil%
  \subfloat[]{\label{fig:se-sefe-sge-example-g2}\includegraphics{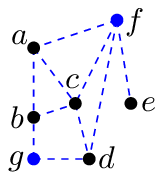}}\hfil%
  \subfloat[]{\label{fig:se-sefe-sge-example-se}\includegraphics{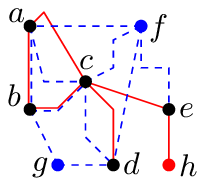}}\hfil%
  \subfloat[]{\label{fig:se-sefe-sge-example-sge}\includegraphics{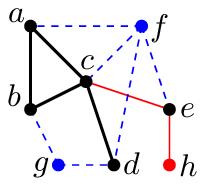}}\hfil%
  \subfloat[]{\label{fig:se-sefe-sge-example-sefe}\includegraphics{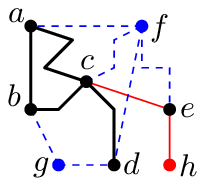}}\hfil%
  \caption{(a-b) $R$ and $B$ with $V_C=\{a,b,c,d,e\}$ and
    $E_C=\{(a,b),(b,c),(a,c),(c,d)\}$. (c) Simultaneous embedding of $R$ and
    $B$. (d) {\sc Sge} of $R$ and $B$. (e) {\sc Sefe} of $R$ and $B$.}
  \label{fig:se-sefe-sge-example}
\end{figure}

Brass~et~al{.}~\cite{bcdeeiklm-spge-07} initiated the research on this topic by investigating \emph{simultaneous geometric embeddings} (or {\sc Sge}s), which are simultaneous embeddings where all edges are represented by straight-line segments (see \figurename~\ref{fig:se-sefe-sge-example-sge}). This setting proved to be fairly restrictive: there exist two trees~\cite{geyer2009two} and even a tree and a path~\cite{angelini2011tree} with no {\sc Sge}. Furthermore, the problem of deciding whether two graphs admit an {\sc Sge} is NP-hard~\cite{estrella2008simultaneous}.

Two relaxations of {\sc Sge} have been considered in the literature in which edges are not forced to be straight-line segments. In the first setting, we look for a simultaneous embedding of two given planar graphs $R$ and $B$ in which every edge is drawn as a polygonal curve with few bends. Erten and Kobourov~\cite{ek-sepgfb-05} proved that three bends per edge always suffice, a bound which has been improved to two bends per edge by Di Giacomo and Liotta~\cite{gl-seogpc-07}. If $R$ and $B$ are trees, then one bend per edge is sufficient~\cite{ek-sepgfb-05}.  Note that black edges may be represented by different curves in each drawing. The variant in which the edges of $R$ and $B$ might only cross at right angles has also been considered~\cite{bdkw-sdpgrac-15}. In the second setting, we look for a {\em simultaneous embedding with fixed edges} (or {\sc Sefe}) of $R$ and $B$: a simultaneous embedding in which every common edge is represented by the same simple curve in the plane (see \figurename~\ref{fig:se-sefe-sge-example-sefe}). In other words, a {\sc Sefe} is a drawing $\Gamma$ of the {\em union graph} $(V_R\cup V_B,E_R\cup E_B)$ such that $\Gamma\vert_R$ 
is a plane drawing of $R$ and $\Gamma\vert_B$ is a plane drawing of $B$. While not every two planar graphs admit a {\sc Sefe}, this setting is substantially less restrictive than {\sc Sge}: for example, every tree and every planar graph admit a {\sc Sefe}~\cite{frati2007embedding}. Determining the complexity of deciding whether two given graphs admit a {\sc Sefe} is a major open problem in the field of graph drawing. Polynomial-time testing algorithms are known in many restricted cases, such as when the common graph $C$ is biconnected~\cite{angelini2012testing}, when $C$ is a set of disjoint cycles~\cite{blasius2015disconnectivity}, or when $R$ is a planar graph and $B$ is a graph with at most one cycle~\cite{fowler2009spqr}. We refer to an excellent survey by Bl\"asius~et~al{.}~\cite{bkr-sepg-13} for many other results.

In this paper we present algorithms to construct {\sc Sefe}s in which edges are represented by polygonal curves. For the purpose of guaranteeing the readability of the representation, we aim at minimizing two natural and well-studied aesthetic criteria in the constructed {\sc Sefe}s: the number of bends per edge and the number of crossings per edge pair. Both criteria have been recently and separately considered in relation to the construction of a {\sc Sefe}. Namely, Grilli et al{.}~\cite{ghkr-dsegfb-14} proved that every combinatorial {\sc Sefe} 
can be realized as a {\sc Sefe} with at most nine bends per edge, a bound which improves to three bends per edge when the common graph is biconnected. Further, Chan et al{.}~\cite{cfglms-dpespg-14} proved that if $R$ and $B$ admit a {\sc Sefe}, then they admit a {\sc Sefe} in which every red-blue edge pair crosses at most twenty-four times.


\subsubsection{Contribution.} In this paper we improve on the results of Grilli et al{.}~\cite{ghkr-dsegfb-14} and of Chan et al{.}~\cite{cfglms-dpespg-14} by proving the following results.

\begin{enumerate}
\item If $R$ and $B$ are both trees, then they admit a {\sc Sefe} with one bend per edge. Consequently, every edge pair crosses at most four times. The number of bends is the best possible, since there exist two trees that do no admit a {\sc Sefe} with no bends~\cite{geyer2009two}. 
\item If $R$ is a planar graph and $B$ is a tree, then they admit a {\sc Sefe} with six bends per edge in which every two exclusive edges cross at most eight times. 
\item If $R$ and $B$ are planar graphs that admit a {\sc Sefe}, then they admit a {\sc Sefe} with six bends per edge in which every two exclusive edges cross at most sixteen times. In all cases, the common edges are drawn as straight-line segments. 
\end{enumerate}

The rest of the paper is organized as follows. In Section~\ref{se:preliminaries} we establish some preliminaries. In Sections~\ref{se:trees},~\ref{se:tree-planar}, and~\ref{se:planar-planar}, we present our results on tree--tree pairs, on tree--planar pairs, and on planar--planar pairs, respectively. Finally, in Section~\ref{se:conclusions} we conclude and suggest some open problems.


\section{Preliminaries}
\label{se:preliminaries}

A {\em plane drawing} of a (multi)graph $G$ is a mapping of each vertex to a point in the plane, and of each edge to a simple curve connecting its endvertices such that no two edges cross. A plane drawing of $G$ determines a circular ordering of the edges incident to each vertex of $G$; the set of these orderings is called a {\em rotation system}. Two plane drawings of $G$ are \emph{equivalent} if they have the same rotation system, the same containment relationship between cycles, and the same outer face (the second condition is redundant if $G$ is connected). A \emph{planar embedding} is an equivalence class of plane drawings. 

Analogously, a {\sc Sefe} of two planar graphs $R$ and $B$ determines a circular ordering of the edges incident to each vertex (comprising edges incident to both $R$ and $B$); the set of these orderings is the {\em rotation system} of the {\sc Sefe}. Two {\sc Sefe}s of $R$ and $B$ are \emph{equivalent} if they have the same rotation system and if their restriction to the vertices and edges of $R$ (of $B$) determines two equivalent plane drawings of $R$ (resp. of $B$). Finally, a \emph{combinatorial {\sc Sefe}} $\E$ for two planar graphs $R$ and $B$ is an equivalence class of {\sc Sefe}s; we denote by $\E\vert_R$ (by $\E\vert_B$) the planar embedding of $R$ (resp. of $B$) obtained by restricting $\E$ to the vertices and edges of $R$ (resp. of $B$).

A {\em subdivision} of a multigraph $G$ is a graph $G'$ obtained by replacing edges of $G$ with paths, whose internal vertices are called {\em subdivision vertices}. If $G'$ is a subdivision of $G$, the operation of {\em flattening} subdivision vertices in $G'$ returns $G$. 
%
The {\em contraction} of an edge $(u,v)$ in a multigraph $G$ leads to a
multigraph $G'$ by replacing $(u,v)$ with a vertex $w$ incident
to all the edges $u$ and $v$ are incident to in $G$; $k$
parallel edges $(u,v)$ in $G$ lead to $k-1$ self-loops
incident to $w$ in $G'$ (the contracted edge itself is not in $G'$).  If $G$ has a planar embedding ${\cal E}_G$, then $G'$ {\em inherits} a planar embedding ${\cal E}_{G'}$ as follows. Let
$a_1,\dots,a_k,v$ and $b_1,\dots,b_\ell,u$ be the clockwise orders of the
neighbors of $u$ and $v$ in ${\cal E}_G$, respectively. 
Then the clockwise order of the neighbors of $w$ is $a_1,\dots,a_k,b_1,\dots,b_\ell$.
The contraction of a connected graph is the contraction of all its edges.




The straight-line segment between points $p$ and $q$ is denoted by $\seg{pq}$. The {\em angle} of $\seg{pq}$ is the angle between the ray from $p$ in positive $x$-direction and the ray from
$p$ through $\seg{pq}$. A polygon $P$ is {\em strictly-convex} if, for any two non-consecutive vertices $p$ and $q$ of $P$, the open segment $\seg{pq}$ lies in the interior of $P$; also, $P$ is {\em star-shaped} if a point $p^*$ exists such that, for any vertex $p$ of $P$, the open segment $\seg{pp^*}$ lies in $P$; the {\em kernel} of $P$ is the set of all such points $p^*$.

A {\em 1-page book embedding} (1PBE) is a plane drawing where all vertices are placed on an oriented line $\ell$ called \emph{spine} and all edges are curves in the halfplane to the left of $\ell$. A {\em 2-page book embedding} (2PBE) is a plane drawing where all vertices are placed on $\ell$ and each edge is a curve in one of the two halfplanes delimited by~$\ell$.


\section{Two Trees} \label{se:trees} In this section we describe an
algorithm that computes a {\sc Sefe} of any two trees $R$ and
$B$ with one bend per edge. Let $C$ be the common graph of $R$ and $B$. 

The outline of the algorithm is as follows. In Step~1, we compute a combinatorial {\sc Sefe} of $R$ and $B$ with the property that at every common vertex $v$, all the black edges are consecutive in the circular order of edges incident to $v$. In Step~2, we contract each component of $C$ to a single vertex, obtaining trees $R'$ from $R$ and $B'$ from $B$.  In Step~3, we independently augment $R'$ and $B'$  to Hamiltonian planar graphs, so as to satisfy topological constraints that are necessary for the subsequent drawing algorithms. In Step~4, we use the Hamiltonian augmentations to construct a simultaneous embedding of $R'$ and $B'$ with one bend per edge; this step is reminiscent of an algorithm of Erten and Kobourov~\cite{ek-sepgfb-05}. Finally, in Step~5, we expand the components of
$C$. This consists of modifying the simultaneous embedding of $R'$ and
$B'$ in a small neighborhood of each vertex to make room for the
components of $C$. We now describe these steps in detail.

\subsubsection{Step~1: Combinatorial Sefe.} Fix the clockwise order of the edges incident to each vertex as follows: all the black edges in any order, then all the red edges in any order, and then all the blue edges in any order (each sequence might be empty). As any rotation system for a tree determines a planar embedding for it, this results in a combinatorial {\sc Sefe} $\E$ of $R$ and $B$. See \figurename~\ref{fig:csefe-csefe}. We may assume that every component $S$ of $C$ is incident to at least one red and one blue edge: If $S$ is not incident to any, say, blue edge, then $B$$=$$S$$=$$C$, since $B$ is connected, and any plane straight-line drawing of $R$ is a {\sc Sefe} of $R$ and $B$.


\begin{figure}[tb]
  \centering%
  \subfloat[]{\label{fig:csefe-csefe}\includegraphics{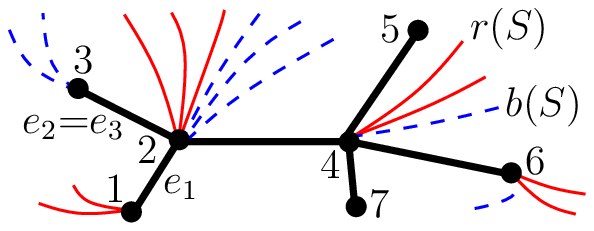}}\hfil%
  \subfloat[]{\label{fig:csefe-redcontracted}\includegraphics{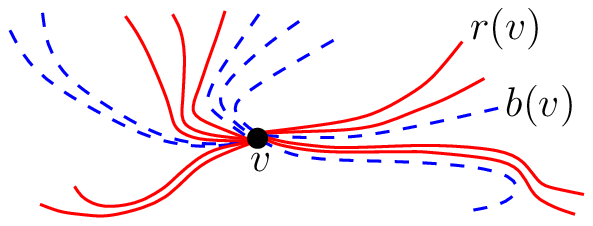}}\hfil%
  \caption{(a) A connected component $S$ of $C$, together with its incident exclusive edges. (b) Vertex $v$ resulting from the contraction of $S$. 
\label{fig:csefe}}
\end{figure}


\remove{
\begin{figure}[b]
  \centering%
  \subfloat[]{\label{fig:csefe-red}\includegraphics[scale=0.51]{Trees-Instance1}}\hfil%
  \subfloat[]{\label{fig:csefe-blue}\includegraphics[scale=0.51]{Trees-Instance2}}\hfil%
  \subfloat[]{\label{fig:csefe-csefe}\includegraphics[scale=0.51]{CombinatorialSEFE}}\hfil%
  \subfloat[]{\label{fig:csefe-redcontracted}\includegraphics[scale=0.51]{CombinatorialSEFE-contracted}}\hfil%
  \caption{(a-b) Trees $R$ and $B$. (c) A combinatorial {\sc Sefe} of
    $R$ and $B$. (d) Trees $R'$ and $B'$ are obtained by contracting the
    connected components of $C$.\label{fig:csefe}}
\end{figure}
}


For every component $S$ of $C$ we pick two incident edges $r(S)$ and $b(S)$ as follows. In any {\sc Sefe} equivalent to $\E$ let $\gamma$ be a simple closed curve surrounding $S$ and close enough to it so that $\gamma$ has no crossing in its interior. Note that $\gamma$ intersects all the exclusive edges incident to $S$ in some clockwise order in which all the exclusive edges incident to a single vertex of $S$ appear consecutively. Let $r(S)$ be any red edge not preceded by a red edge in this order and let $b(S)$ be the first blue edge after $r(S)$. We define a total ordering $\varrho_S$ of the vertices of $S$, as the order in which their exclusive edges intersect $\gamma$ (a curve is added incident to every vertex of $S$ with no incident exclusive edge for this purpose), where the first vertex of $\varrho_S$ is the endvertex of $r(S)$. We have the following. 

\begin{lemma} \label{le:planar-tree-convex-position}
The straight-line drawing of $S$ obtained by placing its vertices on a strictly-convex curve $\lambda$ in the order defined by $\varrho_S$ is plane.
\end{lemma}
\begin{proof}
For every vertex $v$ of $S$, shrink $\gamma$ along an exclusive edge incident to $v$ so that $\gamma$ passes through $v$ and still every edge of $S$ lies in its interior. Eventually $\gamma$ passes through all the vertices of $S$ in the order $\varrho_S$. The planarity of the drawing of $S$ implies that there are no two edges whose endvertices alternate along $\gamma$. Then placing the vertices of $S$ on $\lambda$ in the order $\varrho_S$ leads to a plane straight-line drawing of $S$.
\end{proof}

\subsubsection{Step~2: Contractions.} Contract each component $S$ of $C$ to a
single vertex $v$. The resulting trees $R'$$=$$(V'_R,E'_R)$ and $B'$$=$$(V'_B,E'_B)$ have planar embeddings $\E_{R'}$ and $\E_{B'}$  inherited from $\E_R$ and $\E_B$, respectively. Vertex $v$ is common to $R'$ and $B'$; let $r(v)$ and $b(v)$ be the edges corresponding to $r(S)$ and
$b(S)$ after the contraction. See \figurename~\ref{fig:csefe-redcontracted}.

\subsubsection{Step~3: Hamiltonian augmentations.} We describe this step for $R'$ only;
the treatment of $B'$ is analogous and independent. The goal is to find a vertex order corresponding to a 1PBE of $R'$. All edges between consecutive vertices along the spine $\ell$, as well as the edge between the first and last vertex along $\ell$, can be added to a 1PBE while maintaining planarity, hence the 1PBE is essentially a Hamiltonian augmentation of $R'$. For Step~5 we need to place $r(v)$, for each common vertex $v$, as in the following.
\begin{lemma}\label{lem:be}
  There is a 1PBE for $R'$ equivalent to $\E_{R'}$ such that for every common vertex $v$, the spine passes through $v$ right before $r(v)$ in clockwise order around $v$.
\end{lemma}
\begin{proof}
  We construct the embedding recursively. For each exclusive vertex $v$, let $r(v)$ be an arbitrary edge incident to $v$. Arbitrarily choose a vertex $s$ as the root of $R'$ and place $s$ on $\ell$. Place the other endpoint of $r(s)$ after $s$ on $\ell$ and all remaining neighbors of $s$, if any, in between in the order given by $\E_{R'}$. Then process every child $v$ of $s$ (and the subtree below $v$) recursively as follows  (and ensure that all subtrees stay in pairwise disjoint parts of the spine, for instance, by assigning a specific region to each). 


  \begin{figure}[tb]
    \centering%
    \subfloat[]{\label{fig:recursivebookembedding-a}\includegraphics{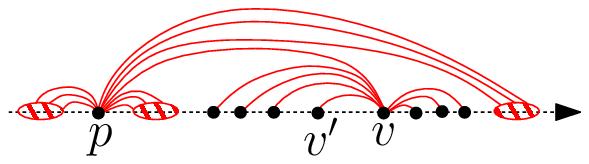}}\hfil%
    \subfloat[]{\label{fig:recursivebookembedding-b}\includegraphics{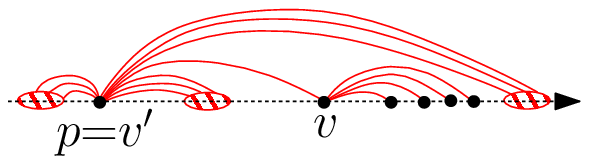}}\hfil%
    \caption{Embedding the children of $v$ if (a) $p\neq v'$ or (b) $p= v'$. Parts of the embedding already constructed are in the shaded regions.\label{fig:recursive-book}}
  \end{figure}

  Note that both $v$ and the parent $p$ of $v$
  are already embedded. By symmetry we can assume that $p$ lies before $v$ on the spine. Let $v'$ be the endvertex of $r(v)$ different from $v$. If $p\neq v'$, we place the other endvertex of $r(v)$ right before $v$. Both if $p\neq v'$ (see \figurename~\ref{fig:recursivebookembedding-a}) and if $p=v'$ (see \figurename~\ref{fig:recursivebookembedding-b}), we place the other children of $v$, if any, according to $\E_{R'}$, in the parts of the spine between $p$ and $v'$, and after $v$. If $v$ is not a leaf, then all its children are processed recursively in the
  same fashion.  It is easily checked that the resulting embedding is a 1PBE that
  satisfies the stated properties.
\end{proof}

\subsubsection{Step~4: Simultaneous embedding.} We now construct a simultaneous embedding of $R'$ and $B'$. In such an embedding let $\sigma_v$ denote the linear order of the edges around each vertex $v$ obtained by sweeping a ray clockwise around $v$, starting in direction of the negative $x$-axis.

\begin{lemma}
\label{lem:tree-tree-se}
For every $\epsilon>0$, $R'$ and $B'$ admit a simultaneous
embedding with one bend per edge in which:

\begin{itemize}
\item all edges of $E_{R'}$ ($E_{B'}$) incident to a vertex $v$ in $V_R'$ (resp.\ $V_B'$) leave $v$ within an angle of $[-\epsilon;+\epsilon]$ with respect to the positive $y$-direction (resp.\ $x$-direction);
\item the drawing restricted to $R'$ (to $B'$) is equivalent to $\E_{R'}$ (resp. to $\E_{B'}$); and
\item for every common vertex $v$, the first red (blue) edge in $\sigma_v$ is $r(v)$ (resp. $b(v)$).
\end{itemize}
\end{lemma}
\begin{proof}
Our algorithm is very similar to algorithms due to Brass et al{.}~\cite{bcdeeiklm-spge-07} and Erten and Kobourov~\cite{ek-sepgfb-05}. These algorithms, however, do not guarantee the construction of a simultaneous embedding in which the order of the edges incident to each vertex is as stated in the lemma. This order is essential for the upcoming expansion step.

We assign the $x$-coordinates $1,\dots,|V_{R'}|$ ($y$-coordinates $|V_{B'}|,\dots,1$) to the vertices of
  $R'$ (resp.\ of $B'$) according to the order in which they occur on the spine in the
  1PBE of $R'$ (resp. of $B'$) computed in Lemma~\ref{lem:be}. This
  determines the placement of every vertex in $V_{R'}\cap V_{B'}$. Set any not-yet-assigned coordinate to $0$.  

We explain how to draw the edges of $R'$: the construction for $B'$ is
  symmetric. The idea is to realize the 1PBE of $R'$ with its
  vertices placed as above and its edges drawn as $x$-monotone
  polygonal curves with one bend. We proceed as follows. The
  1PBE of $R'$ defines a partial order of the edges corresponding to the
  way they nest. For example, denoting the vertices by their order along the spine, edge $(3,4)$ preceeds $(3,5)$ and $(2,5)$, while $(1,2)$ and $(6,7)$ are incomparable. We draw the edges of $R'$ in any linearization of this partial order. Suppose we have drawn some edges and let $(u,v)$ be the next edge to be drawn. Assume w.l.o.g. that the $x$-coordinate of $u$ is smaller than the one of $v$. For some $\epsilon_{uv}>0$, consider the ray
  $\varrho_u$ emanating from $u$ with an angle of $\pi/2-\epsilon_{uv}$
  (with respect to the positive $x$-axis). Similarly, let $\varrho_v$ be
  the ray emanating from $v$ with an angle of $\pi/2+\epsilon_{uv}$.  We
  choose $\epsilon_{uv}<\epsilon$ sufficiently small so that:
  \begin{enumerate}[(1)]\setlength{\itemindent}{2\labelsep}
  \item no vertex in $V_{R'}\setminus\{u\}$ lies in the region to the left
    of the underlying (oriented) line of $\varrho_u$ and to the right of
    the vertical line through $u$;
  \item no vertex in $V_{R'}\setminus\{v\}$ lies in the region to the right
    of the underlying (oriented) line of $\varrho_v$ and to the left of
    the vertical line through $v$; and
  \item neither $\varrho_u$ nor $\varrho_v$ intersects any previously
    drawn edge.
  \end{enumerate}

  As no two vertices of $R'$ have the same $x$-coordinate, we can choose $\epsilon_{uv}$ as claimed. The
  corresponding rays $\varrho_u$ and $\varrho_v$ intersect in some
  point: this is where we place the bend-point of $(u,v)$. The
  resulting drawing is equivalent to the 1PBE of $R'$ and therefore to
  $\E_{R'}$. The remaining claimed properties are preserved from the
  1PBE.
\end{proof}

\subsubsection{Step~5: Expansion.} We now expand the components of $C$ in the
drawing produced by Lemma~\ref{lem:tree-tree-se} one by one in any order. Let $\Gamma$ be the current drawing, $v$ be a vertex corresponding to a not-yet-expanded component $S$ of $C$, and $p$ be the point on which $v$ is placed in $\Gamma$. Note that the red and blue edges incident to $v$ may be
incident to different vertices in $S$. Let $\sigma_v=(e_1,\dots,e_\ell)$, where $e_1,\dots,e_k$ are
red and $e_{k+1},\dots,e_\ell$ are blue. By Lemma~\ref{lem:tree-tree-se}, $r(v)=e_1$ and
$b(v)=e_{k+1}$. Each edge incident to $v$ is drawn as a
polygonal curve with one bend. Let $b_i$ be the bend-point of $e_i$. The plan is to delete $p$ and segments $\seg{pb_i}$ in $\Gamma$ to obtain $\Gamma'$. Then draw $S$ in $\Gamma'$ inside a small disk around $p$ and draw segments from $S$ to $b_1,\dots,b_\ell$. See
\figurename~\ref{fig:tree-tree-expansion}. For an $\epsilon\geq 0$, let $D_\epsilon$ be the disk with
radius $\epsilon$ centered at $p$. Let $\Gamma_R$ ($\Gamma_R'$) be the
restriction of $\Gamma$ (resp.\ $\Gamma'$) to the red and black edges. We
state the following propositions only for the red graph; the 
propositions for the blue graph are analogous. By continuity, $v$ can be moved around slightly in $\Gamma_R$ while maintaining a plane drawing for the red graph. This implies the following.

\begin{figure}
  \centering
  \includegraphics{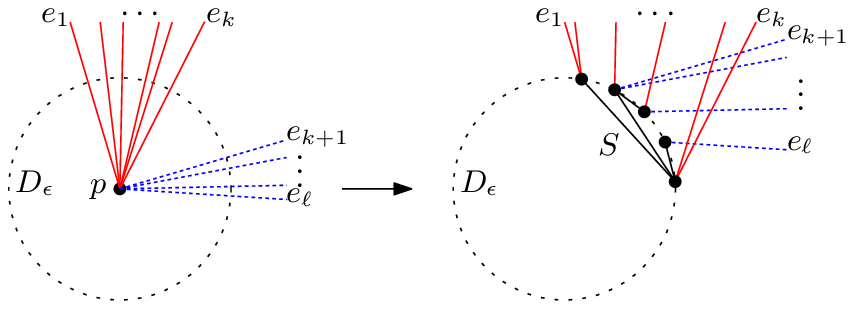}
  \caption{Expanding a component $S$ in a small disk $D_\epsilon$ around
    $p$.}
  \label{fig:tree-tree-expansion}
\end{figure}

\begin{proposition}
  \label{prop:tree-tree-expansion-global}
  There exists a $\delta_R>0$ with the following property. For every
  drawing $\Gamma_R^*$ obtained from $\Gamma_R'$ by drawing $S$ in
  $D_{\delta_R}$, the red segments from $S$ to $b_1,\dots,b_k$ do not
  cross any segment already present in $\Gamma_R'$.
\end{proposition}

The following proposition deals with crossings between red edges incident to $S$.

\begin{proposition}
  \label{prop:tree-tree-expansion-local}
  There exists an $\epsilon_R>0$ with the following property. Let
  $q_1,\dots,q_k$ be any $k$ (not necessarily distinct) points in this
  clockwise order on the upper semicircle of $D_{\epsilon_R}$. Then the
  segments $\seg{q_1b_1},\dots,\seg{q_kb_k}$ do not intersect
  except at common endpoints.
\end{proposition}
\begin{proof}
The angles of $\seg{pb_1},\dots,\seg{pb_k}$ are distinct and strictly decreasing, by Lemma~\ref{lem:tree-tree-se} and by the way $e_1,\dots,e_k$ are labeled. We claim that $\epsilon_R$ can be chosen sufficiently small so that the angles of $\seg{q_1b_1},\dots,\seg{q_kb_k}$ are also distinct and strictly decreasing. For a certain $\epsilon$, let $I_i(\epsilon)$ be the interval of all angles $\alpha$
  such that the ray with angle $\alpha$ from $b_i$ intersects
  $D_{\epsilon}$. Since the angles of $\seg{pb_1},\dots,\seg{pb_k}$ are distinct, it follows that the intervals $I_1(0),\dots,I_k(0)$ are disjoint. By continuity, there
  exists an $\epsilon_R>0$ for which $I_1(\epsilon_R),\dots,I_k(\epsilon_R)$ are also disjoint, and the claim follows for this $\epsilon_R$. Finally, two segments $\seg{q_ib_i}$ and $\seg{q_jb_j}$ with $i<j$ and $q_i\neq
  q_j$ can intersect only if the angle of $\seg{q_ib_i}$ is smaller than
  the angle of $\seg{q_jb_j}$, which does not happen by the claim.
\end{proof}

We get the following main lemma.

\begin{lemma}
  \label{lem:tree-tree-expansion}
  There exists an $\epsilon>0$ with the following property. We can
  expand $S$ to obtain a simultaneous embedding $\Gamma^*$ from
  $\Gamma'$ by drawing the vertices of $S$ on the boundary of
  $D_\epsilon$, the edges of $S$ as straight-line segments, and
  by connecting $S$ to $b_1,\dots,b_\ell$ with straight-line segments.
\end{lemma}
\begin{proof}
  Let $\delta_R$, $\delta_B$, $\epsilon_R$, and $\epsilon_B$ be the
  constants given by Propositions~\ref{prop:tree-tree-expansion-global}
  and~\ref{prop:tree-tree-expansion-local} and their
  analogous formulations for $B$. Let $\epsilon:=\min\{\delta_R,\delta_B,\epsilon_R,\epsilon_B\}$. Place the vertices of $S$ as distinct points on the boundary of the upper-right quadrant of $D_\epsilon$ in the order $\varrho_S$. By Lemma~\ref{le:planar-tree-convex-position}, this placement determines a straight-line plane drawing of $S$. Draw straight-line segments from the vertices of $S$ to $b_1,\dots,b_\ell$, thus completing the drawing of the exclusive edges incident to $S$. We prove that the red segments incident to $S$ do not cross any red or black edge; the proof for the blue segments is analogous. By Proposition~\ref{prop:tree-tree-expansion-global}, the red segments incident to $S$ do not cross the red and black segments not incident to $S$. Also, they do not cross the edges of $S$, which are internal to $D_\epsilon$. Further, Proposition~\ref{prop:tree-tree-expansion-local} ensures that these segments do not cross each other. Namely, the linear order of the vertices of $S$ defined by the sequence of red edges $e_1,\dots,e_k$ is a subsequence of $\varrho_S$, given that the embedding $\E_{R'}$ of $R'$ is the one inherited from $\E_{R}$, given that Lemma~\ref{lem:tree-tree-se} produces a drawing of $R'$ respecting $\E_{R'}$ and in which $e_1=r(v)$, and given that the endvertex of $r(S)$ in $S$ is the first vertex of $\varrho_S$. 
\end{proof}

We are now ready to state the main theorem of this section.

\begin{theorem}
  \label{thm:tree-tree}
  Let $R$ and $B$ be two trees. There exists a {\sc Sefe} of $R$ and $B$ in which every exclusive edge is a polygonal curve with one bend, every common edge is a straight-line segment, and every two exclusive edges cross at most four times.
\end{theorem}
\begin{proof}
  By Lemma~\ref{lem:tree-tree-se}, $R'$ and $B'$ admit a simultaneous embedding with one bend per
  edge. By repeated applications of Lemma~\ref{lem:tree-tree-expansion},
  the simultaneous embedding of $R'$ and $B'$ can be turned into a {\sc Sefe} of $R$ and $B$ in which every exclusive edge has one bend and every common edge is a straight-line segment. Finally, any two exclusive edges cross at most four times, given that each of them   consists of two straight-line segments.
\end{proof}


\section{A Planar Graph and a Tree} \label{se:tree-planar} 

In this section we give an algorithm which computes a {\sc Sefe} of any planar graph $R=(V_R,E_R)$ and any tree $B=(V_B,E_B)$ in which every edge of $R$ has at most six bends and every edge of $B$ has one bend. Due to a possible initial augmentation that maintains planarity, we can assume $R$ to be connected. The common graph $C$ of $R$ and $B$ is a forest, as it is a subgraph of $B$. The algorithm we give has strong similarities with the one for trees (Section~\ref{se:trees}); however, it encounters some of the complications one needs to handle when dealing with pairs of general planar graphs (Section~\ref{se:planar-planar}). Its outline is as follows. 

In Step~1 we modify $R$ and $B$ to obtain a planar graph $R'$ and a tree $B'$ with a common graph $C'$ by introducing {\em antennas}, that are edges of $C'$ replacing parts of the exclusive edges of $R$. While this costs two extra bends per edge of $R$ in the final {\sc Sefe} of $R$ and $B$, it establishes the property that, for every exclusive edge $e$ of $R'$, every endvertex of $e$ in $C'$ is incident to $e$, to an edge of $C'$, and to no other edge. 

In Step~2 we construct a combinatorial {\sc Sefe} of $R'$ and $B'$ such that at every vertex $v$ of $C'$, all the edges of $C'$ are consecutive in the circular order of the edges incident to $v$. While the construction is the same as for tree-tree pairs, it works for the general planar graph $R'$ only because of the antennas introduced in Step~1.

In order to construct a {\sc Sefe} of $R'$ and $B'$, in Steps~3,~5, and~6 we perform a {\em contraction -- simultaneous embedding -- expansion} process similar to the one in Section~\ref{se:trees}. This again relies on an independent Hamiltonian augmentation of the graphs, which is done in Step~4. The augmentation of the tree is done by Lemma~\ref{lem:be}. However, in order to augment the planar graph we need to subdivide some of its edges. Finally, we obtain a {\sc Sefe} of $R$ and $B$ by removing the antennas. Next we describe these steps in detail.


\subsubsection{Step~1: Antennas.} We replace every exclusive edge $e=(u,v)\in E_R$ such that $u,v\in V_R \cap V_B$ by a path $(u,u_e,v_e,v)$, with two new common vertices $u_e$ and $v_e$, black edges $(u,u_e),(v_e,v)$, and a red edge $(u_e,v_e)$. Analogously, we replace every exclusive edge $e=(u,v)\in E_R$ such that $u\in V_R \cap V_B$ and $v\notin V_R \cap V_B$ by a path $(u,u_e,v)$, with a new common vertex $u_e$, a common edge $(u,u_e)$, and a red edge $(u_e,v)$. See \figurename~\ref{fig:degree1}. The resulting planar graph $R'$ and tree $B'$ satisfy the following property.  

\begin{figure}[hb]
  \centering\hfil
  \subfloat[]{\label{fig:degree1-before}\includegraphics{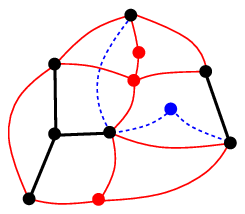}}\hfil%
  \subfloat[]{\label{fig:degree1-after}\includegraphics{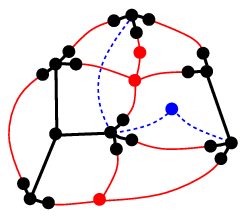}}\hfil
  \caption{(a) Planar graph $R$ and tree $B$. (b) Planar graph $R'$ and tree $B'$.} 
  \label{fig:degree1}
\end{figure}

\begin{prop} \label{pr:degree1}
For every exclusive edge $e$ of $R'$, every endvertex of $e$ in the common graph $C'$ of $R'$ and $B'$ is incident to $e$, to an edge in $C'$, and to no other edge. 
\end{prop}

We also get the following:

\begin{lemma} \label{le:degree1-replacement}
Suppose that a {\sc Sefe} $\Gamma'$ of $R'$ and $B'$ exists in which: (i) every edge of $R'$ (of $B'$) is a polygonal curve with at most $x$ bends (resp. $y$ bends); (ii) every common edge is a straight-line segment; and (iii) any two exclusive edges cross at most $z$ times. Then there exists a {\sc Sefe} $\Gamma$ of $R$ and $B$ in which: (i) every edge of $R$ (of $B$) is a polygonal curve with at most $x+2$ bends (resp. $y$ bends); (ii) every common edge is a straight-line segment; and (iii) any two exclusive edges cross at most $z$ times. 
\end{lemma}
\begin{proof}
  We obtain {\sc Sefe} $\Gamma$ from $\Gamma'$ by removing all edges $(u,u_e)$ and $(v,v_e)$ from the drawing of $B$, and by interpreting all vertices $u_e$ and $v_e$ as bend-points in the drawing of $R$. First, we have that $\Gamma$ is a {\sc Sefe} of $R$ and $B$. In particular, every two edges in $B$ are also edges in $B'$ and since they do not cross in $\Gamma'$, they do not cross in $\Gamma$ either. Further, each edge in $R$ corresponds to a path in $R'$, hence no two edges in $R$ cross in $\Gamma$ as the corresponding paths do not cross in $\Gamma'$. Second, every edge in $C$ is also an edge in $C'$, hence it is a straight-line segment in $\Gamma$, as it is in $\Gamma'$. Third, every edge in $B$ is also an edge in $B'$, hence it is a polygonal curve with at most $y$ bends in $\Gamma$, as it is in $\Gamma'$. Fourth, each edge $e$ in $R$ corresponds to a path in $R'$ composed of at most two edges in $C'$, which are straight-line segments, and of one exclusive edge in $R'$, which has at most $x$ bends. Hence, $e$ has at most $x+2$ bends in $\Gamma$ (the two extra bends correspond to the points where $u_e$ and $v_e$ used to lie). Finally, any exclusive edge in $R$ or $B$ corresponds to at most two edges in $C'$ and of one exclusive edge in $R'$ or $B'$. Since common edges are crossing-free, any two exclusive edges in $R$ and $B$ cross the same number of times as the corresponding exclusive edges in $R'$ and $B'$, which is $z$ by assumption. This concludes the proof.
\end{proof}


\subsubsection{Step~2: Combinatorial Sefe.} Start with any plane drawing of $R'$. This determines the planar embeddings $\E_{R'}$ of $R'$ and $\E_{C'}$ of $C'$. The planar embedding $\E_{B'}$ of $B'$ is completed as for tree-tree pairs: For every vertex $v$ in $C'$, pick a common edge $e_v$ incident to $v$, if it exists. Then draw the exclusive vertices of $B'$ and the blue edges one by one, so that when an edge $(u,v)$ is drawn, it leaves $u$ right after $e_u$ and it enters $v$ right after $e_v$. 
In the resulting combinatorial {\sc Sefe} $\E$ of $R'$ and $B'$ we have, in clockwise order around each vertex of $C'$, either: (i) a (possibly empty) sequence of black edges followed by a (possibly empty) sequence of blue edges; or (ii) a single black edge followed by a single red edge. This is a consequence of Property~\ref{pr:degree1} and of the embedding choice for $B'$. As in Section~\ref{se:trees}, we can assume that every connected component of $C'$ is incident to at least one red and one blue edge. We choose edges $r(S)$ and $b(S)$ for every component $S$ of $C'$ and we define an ordering $\varrho_S$ of the vertices of $S$ as in Section~\ref{se:trees}.

\subsubsection{Step~3: Contractions.} Contract each component of $C'$ to a vertex in $R'$ and in $B'$, determining graphs $R''=(V''_R,E''_R)$ and $B''$, respectively. Note that $R''$ is a planar multigraph, i.e., it might have parallel edges and self-loops, while $B''$ is a tree.  In this way $R''$ inherits a planar embedding $\E_{R''}$ from $\E_{R'}$ and $B''$ inherits a planar embedding $\E_{B''}$ from $\E_{B'}$. Each vertex $v$ resulting from the contraction of a component $S$ of $C'$ is common to $R''$ and $B''$. Let $r(v)$ and $b(v)$ be the edges corresponding to $r(S)$ and $b(S)$ after the contraction.

\subsubsection{Step~4: Hamiltonian augmentations.} A Hamiltonian augmentation of $B''$ is computed by Lemma~\ref{lem:be}. A Hamiltonian augmentation of $R''$ might not exist, thus we subdivide some edges of $R''$ before performing the augmentation, as in the following. 

\begin{lemma} \label{le:graph-augmentation}
There exists a simple planar graph $R'''=(V'''_R,E'''_R)$ such that:

\begin{itemize}
\item {\sc Two-subdivision}: $R'''$ is obtained by subdividing each edge in $R''$ either zero or two times and by adding dummy vertices and edges to the resulting graph;
\item {\sc Embedding}: $R'''$ has a planar embedding ${\cal E}_{R'''}$ from which $\E_{R''}$ can be obtained by removing dummy vertices and edges and flattening subdivision vertices; 
\item {\sc Hamiltonian cycle}: $R'''$ contains a Hamiltonian cycle ${\cal C}$, which we orient counter-clockwise in $\E_{R'''}$, none of whose edges is (part of) an edge in $E''_R$;
\item {\sc Starting edge}: for every common vertex $v$ of $R''$ and $B''$, the edge of ${\cal C}$ entering $v$ comes right before $r(v)$ in the clockwise order of edges incident to $v$ in $\E_{R'''}$; and
\item {\sc Start to the left}: all the edges in $E'''_R$ that are incident to a vertex in $V''_R$ and that are part of an edge in $E''_R$ lie to the left of ${\cal C}$ in ${\cal E}_{R'''}$. 
\end{itemize}
\end{lemma}
\begin{proof}
Let $T$ be a spanning tree of $R''$, which exists since $R''$ is connected. Draw a simple closed curve $\gamma$ in $\E_{R''}$ containing $T$ in its interior and sufficiently close to $T$ so that it crosses every edge in $E''_R$ not in $T$ twice. Insert subdivision vertices for the edges in $E''_R$ not in $T$ at these crossings; also, insert a dummy vertex on $\gamma$ between every two consecutive subdivision vertices of the same edge from $R''$. Orient $\gamma$ counter-clockwise. See \figurename~\ref{fig:graph-augmentation}.a. 

\begin{figure}[tb]
\begin{center}
\begin{tabular}{c c c}
\mbox{\includegraphics{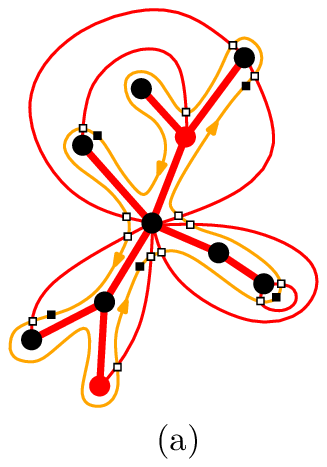}} \hspace{1mm} &
\mbox{\includegraphics{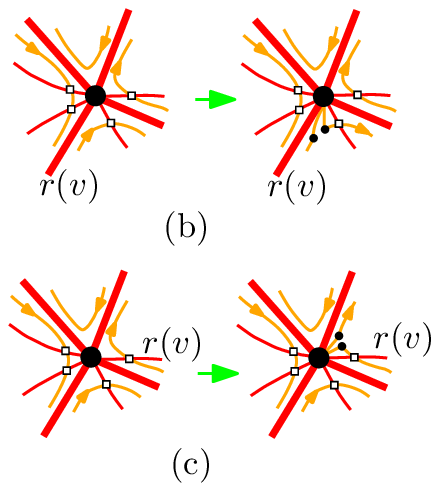}} \hspace{1mm} &
\mbox{\includegraphics{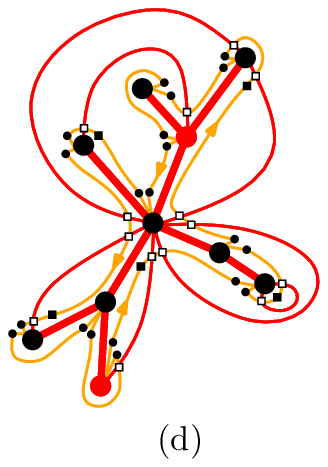}}
\end{tabular}
\caption{Illustration for Lemma~\ref{le:graph-augmentation}. (a) Graph $R''$ with its planar embedding $\E_{R''}$ has red and black disks as vertices and red curves as edges. Curve $\gamma$ is orange. Subdivision vertices for the edges of $R''$ are white squares, while dummy vertices are black squares. (b)--(c) Modifying $\gamma$ so that it passes through a vertex $v$. (d) Graph $R'''$ with its planar embedding $\E_{R'''}$. Both the red and the orange curves represent edges of $R'''$.} 
\label{fig:graph-augmentation}
\end{center}
\end{figure}

For each vertex $v$ that is not a common vertex of $R''$ and $B''$, define $r(v)$ to be an arbitrary edge incident to $v$. We now modify $\gamma$ in a small neighborhood of each vertex $v$ of $R''$, so that $\gamma$ passes through $v$. If $r(v)$ is in $T$, as in \figurename~\ref{fig:graph-augmentation}.b, then while traversing $\gamma$ counter-clockwise stop at a point in which $\gamma$ follows $r(v)$ towards $v$; insert a dummy vertex at that point, then let $\gamma$ take a detour from the dummy vertex to $v$ and then back to its previous route, where another dummy vertex is inserted. If $r(v)$ is not in $T$, as in \figurename~\ref{fig:graph-augmentation}.c, then while traversing $\gamma$ counter-clockwise stop right after the crossing between $\gamma$ and $r(v)$ that is ``closer'' to $v$; insert a dummy vertex on $\gamma$ at that point, then let $\gamma$ take a detour from the dummy vertex to $v$ and then back to its previous route, where another dummy vertex is inserted. Finally, we consider $\gamma$ as a cycle, that is, each curve that is part of $\gamma$ and that connects two consecutive vertices on $\gamma$ is an edge. Denote by $R'''$ the resulting graph and by ${\cal E}_{R'''}$ its planar embedding; see \figurename~\ref{fig:graph-augmentation}.d.

It is easy to verify that $R'''$ and ${\cal E}_{R'''}$ satisfy all the required properties. In particular, the Hamiltonian cycle $\cal C$ required by the statement is the cycle corresponding to $\gamma$: By construction, $\cal C$ passes through every vertex of $R'''$ and it does so right before $r(v)$ in the clockwise order around $v$. Also, every edge of $R''$ has been subdivided twice (if it is not in $T$) or never (if it is in $T$). Further, all the edges in $T$ lie to the left of $\gamma$ and hence of $\cal C$, while all the edges in $R''$ not in $T$ start to the left of $\cal C$, move to its right, and then end again to its left; this implies properties {\sc Hamiltonian cycle} and {\sc Start to the left}. Finally, $R'''$ is simple, due to the introduction of dummy vertices along $\gamma$. 
\end{proof}

\subsubsection{Step~5: Simultaneous embedding.} Ideally, in order to construct a simultaneous embedding of $R''$ and $B''$, we would like to use known algorithms that construct simultaneous embeddings with two bends per edge of every two planar graphs~\cite{ddlw-ccdpg-05,gl-seogpc-07,k-setbepa-06}. However, the existence of self-loops in $R''$ prevents us from doing that. In the following lemma we show how to modify those algorithms to deal with non-simple graphs. \figurename~\ref{fig:planar-tree-se} shows an example of the resulting drawing.



\begin{figure}[tb]
  \centering
  \includegraphics{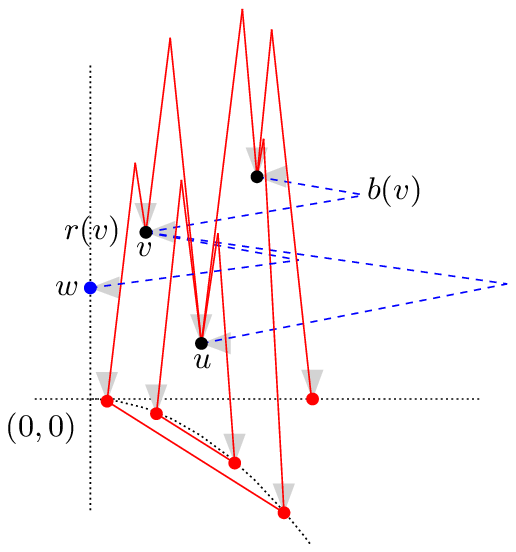}
  \caption{A simultaneous embedding of a planar graph $R''$ and a tree $B''$. The gray cones indicate the angles within $[-\epsilon,+\epsilon]$ of the positive $x$- and $y$-directions. The selfloop at $u$ in $R'$, represented as a path of length three in $R''$, crosses the edge $(v,w)$ eight times. Some angles in the drawing were modified slightly to reduce the height of the figure.}
  \label{fig:planar-tree-se}
\end{figure}

\begin{lemma}
\label{lem:planar-tree-se}
For every $\epsilon>0$, $R''$ and $B''$ admit a simultaneous embedding in which:
\begin{itemize}
\item every edge of $R''$ (of $B''$) is a polygonal curve with at most four bends (resp. with one bend);
\item every two edges cross at most eight times (counting an adjacency as one crossing);
\item all edges of $E_{R''}$ ($E_{B''}$) incident to a vertex $v$ in $V_R''$ (resp.\ $V_B''$) leave $v$ within an angle of $[-\epsilon;+\epsilon]$ with respect to the positive $y$-direction (resp.\ $x$-direction);
\item the drawing restricted to $R''$ (to $B''$) is equivalent to $\E_{R''}$ (resp. to $\E_{B''}$); and
\item for every common vertex $v$, the first red (blue) edge in $\sigma_v$ is $r(v)$ (resp. $b(v)$).
\end{itemize}
\end{lemma}
\begin{proof}
First, we place the vertices of $R'''$ and $B''$ in the plane. Similarly to Lemma~\ref{lem:tree-tree-se}, we assign $x$-coordinates $1,\dots,|V_{R'''}|$ to the vertices of $R'''$ according to the order in which they occur along the Hamiltonian cycle $\cal C$ defined in Lemma~\ref{le:graph-augmentation}, starting at any vertex $u^*$. Further, we assign $y$-coordinates $|V_{B''}|,\dots,1$ to the vertices of $B''$ according to the order in which they occur on the spine in the Hamiltonian augmentation of $B''$ that is computed by Lemma~\ref{lem:be}. This determines the placement of every vertex in $V_{R''}\cap V_{B''}$. We set the $x$-coordinate of every vertex of $B''$ not in $R''$ to $0$; also, we set the $y$-coordinate of every vertex of $R''$ not in $B''$ to $0$. It remains to assign $y$-coordinates to the vertices in $V_{R'''}\setminus V_{R''}$ (note that none of these vertices belongs to $B''$). A subset $V_s$ of the vertices in $V_{R'''}\setminus V_{R''}$ consists of subdivision vertices for the edges in $E_{R''}$; we assign $y$-coordinates to the vertices in $V_s$ so that they lie on the curve $y=-x^2$. We set the $y$-coordinate of every vertex in $V_{R'''}\setminus \{ V_{R''} \cup V_s\}$ to $0$.

We now draw the edges of $R''$ and $B''$. The edges of $B''$ are drawn exactly as in Lemma~\ref{lem:tree-tree-se}. We draw the edges of $R''$ as follows. Note that the Hamiltonian augmentation of $R''$ corresponds to a 2PBE of $R'''$ along a spine $\ell$, where $u^*$ can be assumed w.l.o.g. to be the first vertex along $\ell$. This 2PBE defines a partition of the edges of $R'''$ into those embedded in the half-plane ${\cal H}_l$ to the left of $\ell$ and those embedded in the half-plane ${\cal H}_r$ to the right of $\ell$. Each edge of $R'''$ in ${\cal H}_r$ connects two vertices in $V_s$, which are subdivision vertices for an edge in $E_{R''}$; thus the edges of $R'''$ in ${\cal H}_r$ form a perfect matching on $V_s$. We draw these edges as straight-line segments. In order to draw the edges of $R'''$ in ${\cal H}_l$, we define a partial order $\prec_l$ on these edges, corresponding to the way they are nested. We draw these edges one by one, in an order which is given by any linearization of $\prec_l$. The procedure to draw an edge $(u,v)$ as a $1$-bend edge is the same as in Lemma~\ref{lem:tree-tree-se}. That is, assuming w.l.o.g. that $u$ has $x$-coordinate smaller than $v$, the bend-point is the intersection point between two rays $\varrho_u$ and $\varrho_v$ emanating from $u$ and $v$ with an angle of $\pi/2-\epsilon_{uv}$ and $\pi/2+\epsilon_{uv}$, for some suitably small $0<\epsilon_{uv}<\epsilon$. 

The vertices in $V_{R'''}\setminus \{ V_{R''} \cup V_s\}$ are removed from the drawing, together with their incident edges, while the vertices in $V_s$ are interpreted as bend-points. This determines a drawing $\Gamma_{R''}$ of $R''$. We prove that $\Gamma_{R''}$ and the constructed drawing $\Gamma_{B''}$ of $B''$ constitute a simultaneous embedding of $R''$ and $B''$ as required by the lemma. Drawing $\Gamma_{B''}$ satisfies all the required properties, as in Lemma~\ref{lem:tree-tree-se}. We now argue about $\Gamma_{R''}$.

\begin{itemize}
\item {\em Angle at $v$}: All the edges of $R''$ incident to a vertex $v$ in $V_R''$ leave $v$ within an angle of $[-\epsilon;+\epsilon]$ with respect to the positive $y$-direction, by property {\sc Start to the left} in Lemma~\ref{le:graph-augmentation}, by the fact that edges of $R'''$ to the left of $\cal C$ are in ${\cal H}_l$ in the 2PBE, and by the just described construction for the edges of $R'''$ in ${\cal H}_l$. 
\item {\em Equivalence to $\E_{R''}$}: $\Gamma_{R''}$ is equivalent to $\E_{R''}$ by property {\sc Embedding} in Lemma~\ref{le:graph-augmentation} and since $\E_{R'''}$ determines the 2PBE which the construction of the drawing of $R''$ relies upon. 
\item {\em First edge in $\sigma_v$}: For every common vertex $v$, the first red edge, if any, in $\sigma_v$ is $r(v)$, by property {\sc Starting edge} in Lemma~\ref{le:graph-augmentation}. 
\item {\em Number of bends}: Each edge of $R''$ either coincides with an edge of $R'''$ or consists of three edges of $R'''$, depending on whether it is subdivided zero or two times in the proof of Lemma~\ref{le:graph-augmentation}. If an edge of $R''$ coincides with an edge of $R'''$, then it has one bend in $\Gamma_{R''}$. If it is composed of three edges of $R'''$, then it has four bends in $\Gamma_{R''}$, namely one, zero, and one bend on the three edges of $R'''$ composing it and lying in ${\cal H}_l$, ${\cal H}_r$, and ${\cal H}_l$ in the 2PBE, respectively, plus two bends corresponding to its subdivision vertices. 
\item {\em Planarity}: The vertices in $V_s$ are placed along the convex curve $y=-x^2$ in $\Gamma_{R''}$, in the same order as they occur along $\ell$. Hence, the edges lying in ${\cal H}_r$ in the 2PBE do not cross each other in $\Gamma_{R''}$. That no two edges lying in ${\cal H}_l$ in the 2PBE cross each other in $\Gamma_{R''}$ can be argued as in Lemma~\ref{lem:tree-tree-se}. Finally, any edge lying in ${\cal H}_l$ in the 2PBE has no intersection with the interior of the convex hull of the vertices in $V_s$ (provided that $\epsilon$ is small enough). Hence, it has no intersection in $\Gamma_{R''}$ with any edge lying in ${\cal H}_r$ in the 2PBE.  
\end{itemize}

It remains to argue about the number of crossings between any edges $e_r$ of $R''$ and $e_b$ of $B''$. Note that $e_b$ is composed of two straight-line segments in $\Gamma_{B''}$. If $e_r$ is also composed of two straight-line segments, then $e_r$ and $e_b$ cross at most four times. Otherwise, $e_r$ is composed of five straight-line segments, from which an upper bound of ten on the number of crossings between $e_r$ and $e_b$ directly follows. This bound is improved to eight by observing that the third segment of $e_r$ (corresponding to the edge of $R'''$ lying in ${\cal H}_r$) does not cross the two segments composing $e_b$, as the former lies in the open half-plane $y<0$, while the latter lie in the closed half-plane $y\geq 0$. 
\end{proof}

\subsubsection{Step~6: Expansion.} Next, we expand the components of $C'$ in the simultaneous embedding of $R''$ and $B''$ obtained in Lemma~\ref{lem:planar-tree-se}. This expansion is performed exactly as in Section~\ref{se:trees}. That is, the components of $C'$ are expanded one by one; when a component $S$ is expanded, its vertices are placed in the order $\varrho_S$ on the upper-right quadrant of the boundary of a suitably small disk $D_\epsilon$ centered at the vertex $S$ was contracted to. This results in a {\sc Sefe} of $R'$ and $B'$. Finally, the vertices and edges not in $R$ and $B$ are removed, in order to get a {\sc Sefe} of $R$ and $B$. We have the following.

\begin{theorem}\label{thm:planar-tree}
 Let $R$ be a planar graph and let $B$ be a tree. There exists a {\sc Sefe} of $R$ and $B$ in which every exclusive edge of $R$ is a polygonal curve with at most six bends, every exclusive edge of $B$ is a polygonal curve with one bend, every common edge is a straight-line segment, and every two exclusive edges cross at most eight times.
\end{theorem}
\begin{proof}
By Lemma~\ref{lem:planar-tree-se}, $R''$ and $B''$ admit a simultaneous embedding $\Gamma''$ in which every edge of $R''$ (of $B''$) is a polygonal curve with at most four bends (resp.\ with one bend). By repeated application of Lemma~\ref{lem:tree-tree-expansion}, $\Gamma''$ can be turned into a {\sc Sefe} $\Gamma'$ of $R'$ and $B'$ in which every exclusive edge of $R'$ (of $B'$) has at most four bends (resp.\ one bend) and every common edge is a straight-line segment. By Lemma~\ref{le:degree1-replacement}, there exists a {\sc Sefe} $\Gamma$ of $R$ and $B$ in which every exclusive edge of $R$ (of $B$) has at most six bends (resp.\ one bend) and every common edge is a straight-line segment. Concerning the number of crossings, by Lemma~\ref{lem:planar-tree-se} every two edges cross at most eight times in $\Gamma''$, also counting their adjacencies. While the expansions performed in Lemma~\ref{lem:tree-tree-expansion} in order to construct $\Gamma'$ starting from $\Gamma''$ might introduce new proper crossings for a pair of exclusive edges of $R'$ and $B'$, they only do so at the cost of removing the adjacency between the corresponding edges of $R''$ and $B''$. Hence, the maximum number of crossings per pair of edges is eight in $\Gamma'$ and, by Lemma~\ref{le:degree1-replacement}, is eight also in $\Gamma$. 
\end{proof}

\section{Two Planar Graphs} \label{se:planar-planar} 

In this section we give an algorithm to compute a {\sc Sefe} of any two planar graphs $R$ and $B$ in which every edge has at most six bends. Let $C$ be the common graph of $R$ and $B$. We assume here that a combinatorial {\sc Sefe} $\cal E$ of $R$ and $B$ is part of the input, since testing the existence of a {\sc Sefe} of two planar graphs is a problem of unknown complexity~\cite{bkr-sepg-13}. 

We assume that no exclusive vertex or edge lies in the outer face of $C$ in $\cal E$, and that $R$ and $B$ are connected. These two conditions are indeed met after the following augmentation. First, introduce a cycle $\delta^*$ in $C$ and embed it in $\cal E$ so that it surrounds the rest of $R$ and $B$. Then, introduce a red (blue) vertex inside each face $f$ of $R$ (of $B$) in $\cal E$ different from the outer face, and connect it to all the vertices incident to $f$. 

We outline our algorithm. First, $R$ and $B$ are modified into planar graphs $R'$ and $B'$ with a common graph $C'$ by introducing antennas, as in Step~1 of the algorithm in Section~\ref{se:tree-planar}; however, here the modification is performed for both graphs. Similarly to Sections~\ref{se:trees} and~\ref{se:tree-planar}, we would like to {\em contract} each component $S$ of $C'$ to a vertex, construct a {\em simultaneous embedding} of the resulting graphs, and finally {\em expand} the components of $C'$. However, $S$ is here not just a tree, but rather a planar graph containing in its internal faces other components of $C'$ (and exclusive vertices and edges of $R'$ and $B'$). Hence, the {\em contraction -- simultaneous embedding --
expansion} process does not happen just once, but rather we proceed from the outside to the inside of $C'$ iteratively, each time applying that process to draw certain subgraphs of $R'$ and $B'$, until $R'$ and $B'$ have been entirely drawn. We now describe this algorithm in more detail. 

First, we introduce antennas in $R$ and $B$, that is, we replace each exclusive edge $e=(u,v)$ of $R$ (resp. of $B$) with $u$ and $v$ in $C$ by a path $(u,u_e,v_e,v)$ such that $u_e$, $v_e$, $(u,u_e)$, and $(v_e,v)$ are in $C$, while $(u_e,v_e)$ is exclusive to $R$ (resp.\ to $B$). We also replace each exclusive edge $e=(u,v)$ of $R$ (resp.\ of $B$) with $u$ is in $C$ and $v$ not in $C$ by a path $(u,u_e,v)$ such that $u_e$ and $(u,u_e)$ are in $C$, while $(u_e,v)$ is exclusive to $R$ (resp.\ to $B$). The resulting planar graphs $R'$ and $B'$ satisfy the following property.  

\begin{prop} \label{pr:degree1-bis}
For every exclusive edge $e$, every endvertex of $e$ in the common graph $C'$ of $R'$ and $B'$ is incident to $e$, to an edge in $C'$, and to no other edge. 
\end{prop}

We also get the following lemma, whose proof is analogous to the one of Lemma~\ref{le:degree1-replacement} and hence is omitted here.

\begin{lemma} \label{le:degree1-replacement-bis}
Suppose that a {\sc Sefe} $\Gamma'$ of $R'$ and $B'$ exists in which: (i) every edge of $R'$ (of $B'$) is a polygonal curve with at most $x$ bends (resp. $y$ bends); (ii) every common edge is a straight-line segment; and (iii) any two exclusive edges cross at most $z$ times. Then there exists a {\sc Sefe} $\Gamma$ of $R$ and $B$ in which: (i) every edge of $R$ (of $B$) is a polygonal curve with at most $x+2$ bends (resp. $y+2$ bends); (ii) every common edge is a straight-line segment; and (iii) any two exclusive edges cross at most $z$ times. 
\end{lemma}

A combinatorial {\sc Sefe} $\E'$ of $R'$ and $B'$ is naturally derived from $\E$ by drawing the antennas as ``very small'' curves on top of the edges they partially replace. By Property~\ref{pr:degree1-bis}, in $\E'$ we have, in clockwise order around each vertex of $C'$, either: (i) a sequence of black edges; or (ii) a single black edge followed by a single red edge; or (iii) a single black edge followed by a single blue edge. Let $\E_{C'}$ be the restriction of $\E'$ to $C'$. 


We now construct a {\sc Sefe} of $R'$ and $B'$. We start by representing the cycle $\delta^*$ of $C'$ as a strictly-convex polygon $\Delta^*$. Next, assume that a {\sc Sefe} $\Gamma''$ of two subgraphs $R''$ of $R'$ and $B''$ of $B'$ has been constructed. Let $C''$ be the common graph of $R''$ and $B''$ and let $\E_{R''}$, $\E_{B''}$, and $\E_{C''}$ be the planar embeddings of $R''$, $B''$, and $C''$ in $\E'$, respectively. Assume that the following properties hold for $\Gamma''$. 

\begin{itemize}
\item {\sc Bends and crossings}: every edge of $R''$ or $B''$ is a polygonal curve with at most four bends, every edge of $C''$ is a straight-line segment, and every exclusive edge of $R''$ crosses every exclusive edge of $B''$ at most sixteen times;
\item {\sc Embedding}: the restrictions of $\Gamma''$ to the vertices and edges of $R''$, $B''$, and $C''$ are equivalent to $\E_{R''}$, $\E_{B''}$, and $\E_{C''}$, respectively; and
\item {\sc Polygons}: each not-yet-drawn vertex or edge of $R'$ or $B'$ lies in $\E'$ inside a simple cycle $\delta_f$ in $C''$ which is represented in $\Gamma''$ by a star-shaped empty polygon $\Delta_f$; further, if an edge exists in $C'$ that lies inside $\delta_f$ in $\E'$ and that belongs to the same $2$-connected component of $C'$ as $\delta_f$, then $\Delta_f$ is a strictly-convex polygon. 
\end{itemize}

These properties are initially met with $R''=B''=C''=\delta^*$ and with $\Gamma''=\Delta^*$. In particular, all the vertices and edges of $R'$ and $B'$ that are not part of $\delta^*$ lie inside $\delta^*$ in $\E'$, because of the initial augmentation; further, the interior of $\Delta^*$ in $\Gamma''$ is empty. It remains to describe how to insert in $\Gamma''$ some vertices and edges of $R'$ and $B'$ that are not yet in $\Gamma''$, while maintaining the above properties. Since $R'$ and $B'$ are finite graphs, this will eventually lead to a {\sc sefe} of $R'$ and $B'$. We distinguish two cases.

\paragraph{Case 1:} There exists a simple cycle $\delta_f$ in $C''$ whose interior in $\Gamma''$ is empty, and there exists an edge $e_f$ in $C'$ that lies inside $\delta_f$ in $\E'$ and belongs to the same $2$-connected component of $C'$ as $\delta_f$. By property {\sc Polygons}, $\delta_f$ is represented by a strictly-convex polygon $\Delta_f$ in $\Gamma''$, as in \figurename~\ref{fig:planar-cases-convex}. Consider the maximal $2$-connected subgraph $S_f$ of $C'$ whose outer face in $\E'$ is delimited by $\delta_f$; note that $e_f$ is an edge of $S_f$. As observed in~\cite{hn-acssdpg-10}, a straight-line plane drawing $\Gamma_f$ of $S_f$ exists in which the outer face of $S_f$ is delimited by $\Delta_f$ and every internal face is delimited by a star-shaped polygon. Plug $\Gamma_f$ in $\Gamma''$, so that they coincide along $\Delta_f$, obtaining a drawing $\Gamma'''$, as in \figurename~\ref{fig:planar-cases-skeleton}. Properties {\sc Bends and crossings} and {\sc Embedding} are clearly satisfied by $\Gamma'''$. Concerning property {\sc Polygons}, each vertex or edge of $R'$ or $B'$ that is not in $\Gamma'''$ and that lies inside $\delta_f$ in $\E'$, also lies inside one of the simple cycles $\delta_{f,1},\dots,\delta_{f,k}$ delimiting internal faces of $S_f$, given that these faces partition the interior of $\delta_f$ in $\E'$. Moreover, cycles $\delta_{f,1},\dots,\delta_{f,k}$ are represented by star-shaped polygons, by construction, whose interior is empty in $\Gamma'''$, as the interior of $\Delta_f$ is empty in $\Gamma''$. Finally, no edge exists in $C'$ lying inside $\delta_{f,i}$, for some $1\leq i\leq k$, and belonging to the same $2$-connected component of $C'$ as $\delta_{f,i}$, as any such edge would belong to $S_f$. 

\begin{figure}[tb]
  \centering
  \subfloat[]{\label{fig:planar-cases-convex}\includegraphics{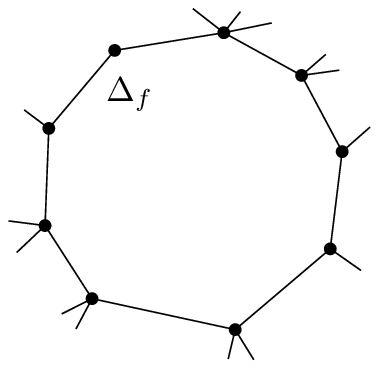}}\hfil%
  \subfloat[]{\label{fig:planar-cases-skeleton}\includegraphics{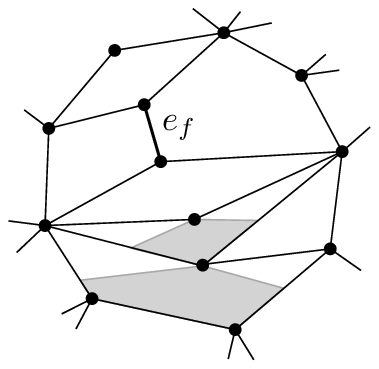}}\hfil%
  \subfloat[]{\label{fig:planar-cases-drawing}\includegraphics{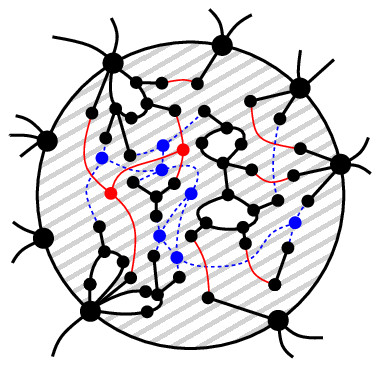}}%
  \caption{(a) Strictly-convex polygon $\Delta_f$ representing $\delta_f$ in $\Gamma''$. (b) Plugging $\Gamma_f$ in $\Gamma''$. The kernels of the star-shaped (non-convex) polygons delimiting internal faces of $S_f$ in $\Gamma_f$ are gray. (c) Graphs $C'_{f}$, $R'_{f}$, and $B'_{f}$; $\delta_f$ is drawn by thick lines and $f$ is shaded.} 
  \label{fig:planar-cases}
\end{figure}

\paragraph{Case 2:} If Case 1 does not apply, then every simple cycle $\delta_f$ in $C''$ whose interior in $\Gamma''$ is empty delimits a face $f$ in its interior in $\E_{C'}$; other vertices and edges of $C'$ might be incident to $f$ though. By property {\sc Polygons}, assuming that $\Gamma''$ is not yet a {\sc Sefe} of $R'$ and $B'$, there exists a simple cycle $\delta_f$ in $C''$ whose interior in $\Gamma''$ is empty containing a not-yet-drawn vertex or edge of $R'$ or $B'$ in its interior in $\E'$, as in \figurename~\ref{fig:planar-cases-drawing}; further, $\delta_f$ is represented by a star-shaped polygon $\Delta_f$ in $\Gamma''$. Let $C'(f)$ be the subgraph of $C'$ composed of the vertices and edges incident to $f$ in $\E_{C'}$. Also, let $R'(f)$ be the subgraph of $R'$ composed of $C'(f)$ and of the red vertices and edges lying in $f$ in $\E'$; graph $B'(f)$ is defined analogously. Let $\E_{R'(f)}$, $\E_{B'(f)}$, and $\E_{C'(f)}$ be the restrictions of $\E'$ to $R'(f)$, $B'(f)$, and $C'(f)$, respectively. We have the following main lemma: 

\begin{lemma} \label{le:one-face-extension}
There exists a {\sc sefe} $\Gamma'_f$ of $R'(f)$ and $B'(f)$ with the following properties:

\begin{itemize}
\item every edge is a polygonal curve with at most four bends, every common edge is a straight-line segment, and every two exclusive edges cross at most sixteen times;
\item $\Gamma'_f$ restricted to $R'(f)$, $B'(f)$, and $C'(f)$ is equivalent to $\E_{R'(f)}$, $\E_{B'(f)}$, and $\E_{C'(f)}$, respectively;
\item cycle $\delta_f$ is represented by $\Delta_f$; and
\item every simple cycle of $C'(f)$ different from $\delta_f$ is represented by an empty strictly-convex polygon in $\Gamma'_f$.
\end{itemize}
\end{lemma}


In order to prove Lemma~\ref{le:one-face-extension}, we present an algorithm consisting of four steps, that resemble Steps 3--6 of the algorithm in Section~\ref{se:tree-planar}. Note that $R'(f)$ and $B'(f)$ are both connected, since $R'$ and $B'$ are connected. We can hence assume that every component $S$ of $C'(f)$ is incident to at least one red and one blue edge, we can choose edges $r(S)$ and $b(S)$, and we can define an ordering $\varrho_S$ of the vertices of $S$ as in Section~\ref{se:trees}.


\subsubsection{Step~1: Contraction.} Contract each component $S$ of $C'(f)$ to a
single vertex $v$. The resulting planar multigraphs $R''(f)$ and $B''(f)$ have  planar embeddings $\E_{R''(f)}$ and $\E_{B''(f)}$ inherited from $\E_{R'(f)}$ and $\E_{B'(f)}$. Vertex $v$ is common to $R''(f)$ and $B''(f)$. Let $r(v)$ and $b(v)$ be the edges corresponding to $r(S)$ and $b(S)$ after the contraction. We stress the fact that a component $S^*$ of $C'(f)$ contains cycle $\delta_f$. While $S^*$ is contracted to a vertex $u^*$ as every other component of $C'(f)$, it will later play a special role. We also remark that, unlike the other components, the order of the edges incident to $S^*$ ``changes'' after the contraction. That is, consider $\E_{R'(f)}$ and $\E_{B'(f)}$, draw a simple closed curve $\gamma$ in the interior of $f$ arbitrarily close to $S^*$, and consider the {\em counter-clockwise} order in which the edges of $R'(f)$ and $B'(f)$ intersect $\gamma$; then this is also the {\em clockwise} order in which the same edges are incident to $u^*$ in $\E_{R''(f)}$ and $\E_{B''(f)}$. 

\subsubsection{Step~2: Hamiltonian augmentations.} We compute Hamiltonian augmentations $R'''(f)$ of $R''(f)$ and $B'''(f)$ of $B''(f)$. This is done independently for $R''(f)$ and $B''(f)$ by means of the algorithm in the proof of Lemma~\ref{le:graph-augmentation} (with an obvious mapping between the notation of Lemma~\ref{le:graph-augmentation} and the one here). As opposed to Section~\ref{se:tree-planar}, both graphs might need to be subdivided in order to augment them to Hamiltonian. 

\subsubsection{Step~3: Simultaneous embedding.} A simultaneous embedding of $R''(f)$ and $B''(f)$ is constructed by means of an algorithm very similar to the one in the proof of Lemma~\ref{lem:planar-tree-se}. Let $\sigma_v$ be defined as in Section~\ref{se:trees}. We have the following. 

\begin{lemma} \label{lem:planar-planar-se}
For every $\epsilon>0$, $R''(f)$ and $B''(f)$ admit a simultaneous embedding $\Gamma''_f$ in which:
\begin{itemize}
\item every edge of $R''(f)$ and $B''(f)$ is a polygonal curve with at most four bends;
\item every two edges cross at most sixteen times (counting an adjacency as one crossing);
\item all edges of $R''(f)$ ($B''(f)$) incident to a vertex $v$ in $R''(f)$ ($B''(f)$) leave $v$ within an angle of $[-\epsilon;+\epsilon]$ with respect to the positive $y$-direction (resp.\ $x$-direction);
\item $\Gamma''_f$ restricted to $R''(f)$ ($B''(f)$) is equivalent to $\E_{R''(f)}$ (resp. $\E_{B''(f)}$);
\item for every common vertex $v$, the first red (blue) edge in $\sigma_v$ is $r(v)$ (resp. $b(v)$); and 
\item vertex $u^*$ is at point $(1,|V(B'''(f))|)$; the straight-line segments incident to $u^*$ in the drawing of $R''(f)$ (of $B''(f)$) have their endpoints different from $u^*$ on the straight line $x=1.5$ (resp. $y=|V(B'''(f))|-0.5$); every other vertex or bend of an edge of $R''(f)$ (resp. of $B''(f)$) is to the right of (resp. below) that line.  
\end{itemize}
\end{lemma}
\begin{proof}
The algorithm to draw $R''(f)$ and $B''(f)$ is similar to one presented in the proof of Lemma~\ref{lem:planar-tree-se} to draw $R''$. Assign the vertices of $R'''(f)$ (of $B'''(f)$) with distinct positive integer $x$-coordinates (resp. $y$-coordinates) according to their order in the Hamiltonian cycle of $R'''(f)$ (resp. according to the reverse order in the Hamiltonian cycle of $B'''(f)$). It is important here that  $u^*$ is the vertex of $R'''(f)$ (of $B'''(f)$) that gets the smallest $x$-coordinate (resp. the largest $y$-coordinate). Place the subdivision vertices for the edges of $R''(f)$ (of $B''(f)$) on the curve $y=-x^2$ (resp. $x=-y^2$); set any non-assigned coordinate to $0$. The edges of $R'''(f)$ and $B'''(f)$ are drawn as the edges of $R'''$ in the proof of Lemma~\ref{lem:planar-tree-se}, except for the edges incident to $u^*$. Namely, the bend-point of an edge $(u^*,v)$ of $R'''(f)$ (of $B'''(f)$) is placed at the intersection point between the line $x=1.5$ (resp. $y=|V(B'''(f))|-0.5$) and the ray $\varrho_v$ emanating from $v$ with an angle of $\pi/2+\epsilon_{{u^*}v}$ (resp. $\epsilon_{{u^*}v}$), for some suitably small $0<\epsilon_{{u^*}v}<\epsilon$.



Remove from the drawing the vertices of $R'''(f)$ and $B'''(f)$ that are neither vertices of $R''(f)$ or $B''(f)$, nor subdivision vertices for the edges of $R''(f)$ or $B''(f)$, and interpret the subdivision vertices for the edges of $R''(f)$ and $B''(f)$ as bend-points. This results in a {\sc Sefe} $\Gamma''_f$ of $R''(f)$ and $B''(f)$, which can be proved to satisfy the required properties exactly as in the proof of Lemma~\ref{lem:planar-tree-se}. In particular, if edges $e_r$ of $R''(f)$ and $e_b$ of $B''(f)$ have four bends each, then they cross at most twenty-five times. This bound can be improved to sixteen, since the third segment of $e_r$ does not cross any of the five segments composing $e_b$ (given that the former lies in the open half-plane $y<0$, while the latter lie in the closed half-plane $y\geq 0$), and vice versa. 
\end{proof}

\subsubsection{Step~4: Expansion.} This step is more involved than in Sections~\ref{se:trees} and~\ref{se:tree-planar}, because when expanding the component $S^*$, we need to ensure that the cycle $\delta_f$ is drawn as $\Delta_f$. 

We first expand the components $S\neq S^*$ of $C'(f)$ in $\Gamma''_f$. Differently from the previous sections, $S$ is a {\em cactus graph}, rather than a tree, which is a graph whose vertices and edges are all incident to a common face, in this case $f$. However, Lemma~\ref{le:planar-tree-convex-position} holds true (with the same proof) even if $S$ is a cactus graph. Hence, we expand the components $S\neq S^*$ one by one in $\Gamma''_f$. When a component $S$ is expanded, its vertices are placed in the order $\varrho_S$ on the upper-right quadrant of the boundary of a suitably small disk $D_\epsilon$ centered at the vertex $S$ was contracted to. Denote again by $\Gamma''_f$ the resulting {\sc Sefe} in which every component $S\neq S^*$ of $C'(f)$ has been expanded. Note that every simple cycle of each component $S\neq S^*$ is an empty strictly-convex polygon in $\Gamma''_f$, since its incident vertices lie on a strictly-convex curve, namely the boundary of $D_\epsilon$.

\begin{figure}[tb]
  \centering
  \subfloat[]{\label{fig:planar-cases-2-a}\includegraphics{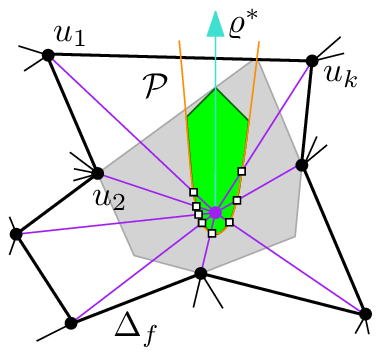}}\hfil%
  \subfloat[]{\label{fig:planar-cases-2-b}\includegraphics{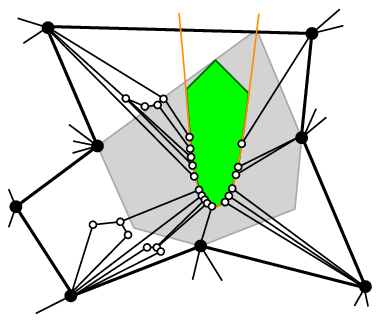}}\hfil%
  \subfloat[]{\label{fig:planar-cases-2-d}\includegraphics{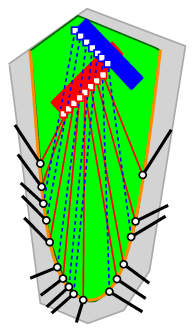}}
\caption{(a) Geometry inside $\Delta_f$: $p^*$ is purple, ${\cal H}^*$ is green, the kernel of $\Delta_f$ is gray, and points $p_1,\dots,p_k$ are empty squares. (b) Drawing $\Gamma^*$; vertices of $S^*$ not in $\delta_f$ are white disks. (c) Reconnecting $\Gamma''_f$ to $\Gamma^*$; drawing $\Gamma''_f$ is represented by a red and a blue rectangle. Red and blue squares represent bendpoints of $R'''(f)$ and $B'''(f)$ adjacent to $u^*$.} 
\label{fig:planar-cases-2}
\end{figure}

In order to complete the construction of a {\sc Sefe} $\Gamma'_f$ of $R'(f)$ and $B'(f)$ as requested by Lemma~\ref{le:one-face-extension}, it remains to deal with the cactus graph $S^*$ containing $\delta_f$ whose edges are all incident to $f$. We sketch the plan: define a region ${\cal H}^*$ inside $\Delta_f$ (\figurename~\ref{fig:planar-cases-2-a}); construct a drawing $\Gamma^*$ of $S^*$ such that $\delta_f$ is represented as $\Delta_f$ and all the other vertices and edges of $S^*$ are inside $\Delta_f$ but outside ${\cal H}^*$ (\figurename~\ref{fig:planar-cases-2-b}); rotate and scale $\Gamma''_f$ and place it in ${\cal H}^*$; finally, connect $\Gamma^*$ with $\Gamma''_f$ via straight-line segments, thus obtaining $\Gamma'_f$ (\figurename~\ref{fig:planar-cases-2-d}).

We begin by defining ${\cal H}^*$. Denote by $u_1,\dots,u_k$ the counter-clockwise order of the vertices along $\delta_f$. Removing the edges of $\delta_f$ disconnects $S^*$ into $k$ cactus graphs; w.l.o.g. assume that $u_1$ is in the one of these cactus graphs that is incident to $r(S^*)$. By property {\sc Polygons}, $\Delta_f$ is star-shaped, hence it has a non-empty kernel. Let $p^*$ be any point in this kernel and $\varrho^*$ be a ray emanating from $p^*$ through $u_1$. Rotate $\varrho^*$ clockwise around $p^*$ of a sufficiently small angle so that no vertex of $\Delta_f$ is encountered during the rotation. For sake of simplicity of description, assume that the origin of the Cartesian axes is at $p^*$, with $\varrho^*$ being the positive $y$-axis. Draw a parabola $\cal P$ with equation $y=ax^2-b$, with $a,b>0$; $a$ is large enough and $b$ is small enough so that $\cal P$ intersects all of $\seg{p^*u_1},\dots,\seg{p^*u_k}$ at points $p_1,\dots,p_k$, and so that a wedge ${\cal W}^*$ with angle $\pi/2$, centered at a point on $\varrho^*$, and bisected by a ray in the negative $y$-direction exists containing all of $p_1,\dots,p_k$ in its interior and having an intersection ${\cal H}^*$ with the region $y>ax^2-b$ entirely lying in the kernel of $\Delta_f$. Let be ${\cal P}^*$ the part of $\cal P$ in ${\cal W}^*$.

Next, we construct a drawing $\Gamma^*$ of $S^*$ (see \figurename~\ref{fig:planar-cases-2-b}).

\begin{lemma} \label{lem:drawing-s-star}
There exists a straight-line plane drawing $\Gamma^*$ of $S^*$ such that: (i) $\Gamma^*$ is equivalent to the restriction of $\E_{C'(f)}$ to $S^*$; (ii) $\delta_f$ is represented by $\Delta_f$; (iii) every simple cycle of $S^*$ different from $\delta_f$ is an empty strictly-convex polygon in $\Gamma^*$; (iv) all the vertices of $S^*$ incident to exclusive edges of $R'(f)$ and $B'(f)$ are on ${\cal P}^*$ and in the interior of ${\cal W}^*$; and (v) $\Gamma^*$ has no intersection with ${\cal H}^*$.
\end{lemma} 
\begin{proof}
We construct $\Gamma^*$ by iteratively drawing $2$-connected components of $S^*$; every such component is either a simple cycle or an edge, since $S^*$ is a cactus graph. Recall that, by Property~\ref{pr:degree1-bis}, every vertex of $S^*$ incident to an exclusive edge of $R'(f)$ or $B'(f)$ has degree one in $S^*$. 

Initialize $\Gamma^*$ by drawing straight-line segments from $u_1,\dots,u_k$ to points on ${\cal P}^*$. For each $1\leq i\leq k$, the number of drawn straight-line segments incident to $u_i$ is equal to the number of $2$-connected components of $S^*$ containing $u_i$ and different from $\delta_f$; by choosing the endpoints of the segments incident to $u_i$ sufficiently close to $p_i$ on ${\cal P}^*$, it can be ensured that all these segments do not cross each other and have empty intersection with ${\cal H}^*$. Some drawn segments are ``real'', that is, they represent edges of $S^*$. Some other segments are ``dummy'', that is, they represent subgraphs of $S^*$ that still need to be drawn. Straight-line segments appear around each vertex $u_i$ in the order in which the corresponding subgraphs of $S^*$ appear around $u_i$ according to $\E_{C'(f)}$.

Now assume to have a plane straight-line drawing $\Gamma^*$ of a subgraph $D^*$ of $S^*$ such that the following invariant is satisfied (in addition to the properties in the statement of the lemma).
\begin{quote}
Consider the cactus graphs that result from the removal of the edges of $D^*$ from $S^*$. Each of these graphs that is not a single vertex is represented in $\Gamma^*$ by a dummy straight-line segment from its only vertex in $\Gamma^*$ to a point on ${\cal P}^*$; further, all these dummy straight-line segments do not cross each other, do not cross any other segment in $\Gamma^*$, and have empty intersection with ${\cal H}^*$.  
\end{quote}

\begin{figure}[tb]
\begin{center}
\begin{tabular}{c c c}
\mbox{\includegraphics{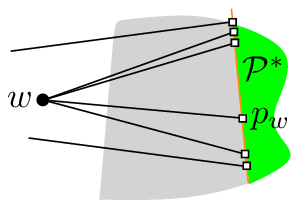}} \hspace{2mm} &
\mbox{\includegraphics{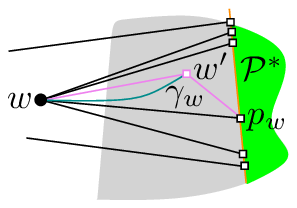}} \hspace{2mm} &
\mbox{\includegraphics{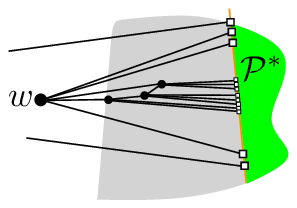}}\\
(a) \hspace{2mm}& (b) \hspace{2mm}& (c)
\end{tabular}
\caption{(a) A dummy straight-line segment $\seg{wp_w}$ representing a component $D$ in $\Gamma^*$. (b) Point $w'$ and curve $\gamma_w$. (c) Drawing component $D$ in $\Gamma^*$.} 
\label{fig:planar-cases-3}
\end{center}
\end{figure}

Note that the invariant is satisfied by $\Gamma^*$ after the initialization. Then it suffices to show how the invariant is maintained after drawing in $\Gamma^*$ a $2$-connected component $D$ of $S^*$, where just one vertex $w$ of $D$ is already in $\Gamma^*$, and where a dummy straight-line segment $\seg{wp_w}$ with $p_w\in {\cal P}^*$ represents $D$ in $\Gamma^*$, as in \figurename~\ref{fig:planar-cases-3}.a. Observe that $D$ is a simple cycle, as if it were an edge, it would be represented by a real straight-line segment and not a dummy straight-line segment. Consider a point $w'$ arbitrarily close to the midpoint of $\seg{wp_w}$. Draw a strictly-convex curve $\gamma_w$ inside triangle $\Delta_w=(w,w',p_w)$ connecting $w$ and $w'$, as in \figurename~\ref{fig:planar-cases-3}.b. Place the vertices of $D$ on $\gamma_w$, in the order they occur along $D$ according to $\E_{C'(f)}$; also, draw the edges of $D$ as straight-line segments, as in \figurename~\ref{fig:planar-cases-3}.c. Remove $\seg{wp_w}$ from $\Gamma^*$.  The polygon representing $D$ is empty, provided that $w'$ is sufficiently close to $\seg{wp_w}$, and strictly-convex, since its incident vertices lie on a strictly-convex curve. Further, the straight-line segments from the vertices of $D$ to $p_w$ do not cross each other, do not cross any other segment in $\Gamma^*$, and have empty intersection with ${\cal H}^*$, provided that $w'$ is sufficiently close to $\seg{wp_w}$. Hence, a suitable number of points on ${\cal P}^*$ can be chosen, all sufficiently close to $p_w$ so that the straight-line segments between these points and the vertices of $D$ do not cross each other, do not cross any other segment in $\Gamma^*$, and have empty intersection with ${\cal H}^*$; thus, the invariant is satisfied by the new $\Gamma^*$, which concludes the proof. 
\end{proof}


The construction of $\Gamma'_f$ is completed as follows (see \figurename~\ref{fig:planar-cases-2-d}). First, we delete $u^*$ and its incident straight-line segments from $\Gamma''_f$. Second, we rotate $\Gamma''_f$ counter-clockwise by an angle of $3\pi/4$. Third, we scale $\Gamma''_f$ down so that it fits inside a disk $D_{\varepsilon}$ with a suitably small radius $\varepsilon>0$. Fourth, we place $\Gamma''_f$ in $\Gamma^*$ so that $D_{\varepsilon}$ is inside ${\cal H}^*$ and is tangent to the half-lines delimiting ${\cal W}^*$. Finally, we complete the drawing of the exclusive edges of $R'(f)$ and $B'(f)$ by drawing straight-line segments from their bend-points previously adjacent to $u^*$ to the suitable vertices of $S^*$ on ${\cal P}^*$. We have the following.

\begin{lemma} \label{lem:complete-drawing}
$\Gamma'_f$ is a {\sc Sefe} of $R'(f)$ and $B'(f)$ with the properties required by Lemma~\ref{le:one-face-extension}, provided that $\varepsilon$ is sufficiently small. 
\end{lemma}
\begin{proof}
We first prove that $\Gamma'_f$ is a {\sc Sefe} of $R'(f)$ and $B'(f)$. In particular, vertices and edges of $C'(f)$ have a unique representation in $\Gamma'_f$, hence it suffices to prove that the drawings of $R'(f)$ and $B'_f$ in $\Gamma'_f$ are planar; we will argue about the planarity of the drawing of $R'(f)$, as the one of $B'_f$ can be proved analogously. 

By Lemma~\ref{lem:planar-planar-se}, the drawing of $R''(f)$ in $\Gamma''_f$ is plane. By Lemma~\ref{le:planar-tree-convex-position}, $\Gamma''_f$ stays plane after all the components different from $S^*$ have been expanded. By Lemma~\ref{lem:drawing-s-star}, the drawing $\Gamma^*$ of $S^*$ is plane, as well. Further, $\Gamma''_f$ and $\Gamma^*$ do not cross each other, as the former lies in a disk $D_{\varepsilon}$ which is inside ${\cal H}^*$, provided that $\varepsilon$ is sufficiently small, while the latter does not intersect ${\cal H}^*$, by Lemma~\ref{lem:drawing-s-star}. It remains to argue that the straight-line segments drawn to restore the exclusive edges of $R'(f)$ do not cause crossings. 
\begin{itemize}
\item First, these segments lie in ${\cal H}^*$ if $\varepsilon$ is small enough, hence they do not intersect $\Gamma^*$. 
\item Second, they do not intersect red edges in $\Gamma''_f$; namely, by Lemma~\ref{lem:planar-planar-se} and assuming that the components of $C'(f)$ different from $S^*$ have been expanded in sufficiently small disks, we have that all the red edges in $\Gamma''_f$ lie to the right of the line $\ell_v$ with equation $x=1.5$. After the rotation of $\Gamma''_f$ by $3\pi/4$ counter-clockwise, $\Gamma''_f$ is above the rotated line $\ell_v$. Thus, it suffices to prove that all the straight-line segments drawn to reconnect the exclusive edges of $R'(f)$ are below or on $\ell_v$. Indeed, by Lemmata~\ref{lem:planar-planar-se} and~\ref{lem:drawing-s-star} each of these segments has one endpoint on $\ell_v$ and the other endpoint in the interior of ${\cal W}^*$; further, $\ell_v$ is arbitrarily close, depending on the value of $\varepsilon$, to the line delimiting ${\cal W}^*$ with slope $5\pi/4$. Hence, each straight-line segment drawn to reconnect an exclusive edge of $R'(f)$ has one end-point on $\ell_v$ and one end-point below it, provided that $\varepsilon$ is sufficiently small. 
\item Third, the straight-line segments drawn to reconnect the exclusive edges of $R'(f)$ do not cross each other, since the clockwise order in which the edges of $R'(f)$ are incident to $u^*$ (which by Lemma~\ref{lem:planar-planar-se} is also the left-to-right order in which the endpoints of the deleted red straight-line segments appear on $\ell_v$ after the rotation) coincides with the counter-clockwise order in which they are incident to vertices in $S^*$ (which is also the left-to-right order in which these vertices appear along ${\cal P}^*$). 
\end{itemize}

The bound on the number of bends in $\Gamma'_f$ follows from the corresponding bound for $\Gamma''_f$ in Lemma~\ref{lem:planar-planar-se} and from the fact that, when a component of $C'(f)$ is expanded, no new bends are introduced on the exclusive edges. In particular, the exclusive edges incident to one or two vertices in $S^*$ have respectively one or two straight-line segments in $\Gamma'_f$ they did not have in $\Gamma''_f$; however, in turn they lost respectively one or two straight-line segments in $\Gamma'_f$ they used to have in $\Gamma''_f$, namely those incident to $u^*$. 

The bound on the number of crossings is established as in Theorem~\ref{thm:planar-tree}. Consider any two exclusive edges $e'_r$ of $R'(f)$ and $e'_b$ of $B'(f)$. By Lemma~\ref{lem:planar-planar-se}, the corresponding edges $e''_r$ in $R''(f)$ and  $e''_b$ in $B''(f)$ cross at most sixteen times in $\Gamma''_f$, also counting adjacencies. While the expansions might introduce proper crossings between $e'_r$ and $e'_b$, they only do so in correspondence of an adjacency between $e''_r$ and $e''_b$; hence $e'_r$ and $e'_b$ cross at most sixteen times in $\Gamma'_f$. 

Finally, the properties that the edges of $C'(f)$ are straight, that $\Gamma'_f$ restricted to $R'(f)$, $B'(f)$, and $C'(f)$ is equivalent to $\E_{R'(f)}$, $\E_{B'(f)}$, and $\E_{C'(f)}$, respectively, that $\delta_f$ is represented by $\Delta_f$, and that every simple cycle of $C'(f)$ different from $\delta_f$ is represented in $\Gamma'_f$ by an empty strictly-convex polygon have been explicitly ensured while performing the construction. 
\end{proof}

Lemma~\ref{lem:complete-drawing} concludes the proof of Lemma~\ref{le:one-face-extension}. Next, plug $\Gamma'_f$ in $\Gamma''$, so that they coincide along $\Delta_f$, obtaining a drawing $\Gamma'''$. As in Case~1 and relying on Lemma~\ref{le:one-face-extension}, it is easily shown that Properties {\sc Bends and crossings}, {\sc Embedding}, and {\sc Polygons} are satisfied by $\Gamma'''$, thus completing the discussion of Case~2. We get the following.

\begin{theorem}\label{thm:planar-planar}
Let $R$ and $B$ be two planar graphs. If there exists a {\sc Sefe} of $R$ and $B$, then there also exists a  {\sc Sefe} in which every edge is a polygonal curve with at most six bends, every common edge is a straight-line segment, and every two exclusive edges cross at most sixteen times.
\end{theorem}
\begin{proof}
By property {\sc Bends and crossings}, every drawing $\Gamma''$ constructed by initializing $R''=B''=C''=\delta^*$ and $\Gamma''=\Delta^*$, and by then repeatedly applying Case~1 or Case~2 described above is such that every exclusive edge is a polygonal curve with at most four bends, every common edge is a straight-line segment, and every two exclusive edges cross at most sixteen times. Eventually $\Gamma''=\Gamma'$ is a {\sc Sefe} of $R'$ and $B'$. By Lemma~\ref{le:degree1-replacement-bis}, the drawing obtained from $\Gamma'$ by removing vertices and edges not in $R$ and $B$ is a {\sc Sefe} of $R$ and $B$ satisfying the required properties.
\end{proof}

\section{Conclusions} \label{se:conclusions}

In this paper we proved upper bounds for the number of bends per edge and the number of crossings required to realize a {\sc Sefe} with polygonal curves as edges. While the bound on the number of bends per edge we presented for tree-tree pairs is tight, there is room for improvement for pairs of planar graphs, as the best known lower bound~\cite{bcdeeiklm-spge-07} only states that one bend per edge might be needed. We suspect that our upper bound could be improved by designing an algorithm that constructs a simultaneous embedding of two planar multigraphs with less than four bends per edge.  A related interesting problem is to determine how many bends per edge are needed to construct a simultaneous embedding (without fixed edges) of pairs of (simple) planar graphs. The best known upper bound is two~\cite{ddlw-ccdpg-05,gl-seogpc-07,k-setbepa-06} and the best known lower bound is one~\cite{fkk-csnse-09}. As a final research direction, we mention the problem of constructing {\sc Sefe}s of pairs of planar graphs in polynomial area, while matching our bounds for the number of bends and crossings.

\subsubsection*{Acknowledgments} This research initiated at the Workshop on Geometry and Graphs, held at the Bellairs Research Institute in Barbados in March 2015. The authors thank the other participants for a stimulating atmosphere. Frati also wishes to thank Anna Lubiw and Marcus Schaefer for insightful ideas they shared during the research for~\cite{cfglms-dpespg-14}.


\end{document}